\documentclass[12pt,a4paper]{amsart}

\usepackage{amscd,amssymb,amsthm}
\usepackage{graphicx}
\usepackage[utf8]{inputenc}
\usepackage[english]{babel}
\usepackage{color}

\numberwithin{equation}{section}

\theoremstyle{plain}
\newtheorem{theorem}{Theorem}[section]

\theoremstyle{definition}

\theoremstyle{remark}

\numberwithin{equation}{section}

\numberwithin{table}{section}

\numberwithin{figure}{section}

\newcommand{\bea}{\begin{eqnarray*}}
\newcommand{\eea}{\end{eqnarray*}}
\newcommand{\bean}{\begin{eqnarray}}
\newcommand{\eean}{\end{eqnarray}}

\def\Gf#1,#2,#3,#4;{G^{#1,#2}_{#3,#4}}
\def\Hf#1,#2,#3,#4;{H^{#1,#2}_{#3,#4}}

\def\gfi#1,#2,#3,#4;{I^{(#1_i),(#2_i)}_{(#3_i),#4}}
\def\gfd#1,#2,#3,#4;{D^{(#1_i),(#2_i)}_{(#3_i),#4}}

\def\phi{\varphi}
\begin{document}
\title[Fractional supersymmetric quantum mechanics and lacunary
 Hermite polynomials]
{Fractional supersymmetric quantum mechanics  and lacunary Hermite polynomials}
\author{F. Bouzeffour $^\ast$,  M. Garayev $^\diamond$}
\address{$^\ast$ Department of mathematics, College of Sciences, King Saud University,  P. O Box 2455 Riyadh 11451, Saudi Arabia and  College of Sciences, University of Carthage, B. O. Box 64 Jarzouna 7021, Bizerte, Tunisia.}
\address{$^\diamond$ Department of mathematics, College of Sciences, King Saud University,  P. O Box 2455 Riyadh 11451, Saudi Arabia and  College of Sciences}
\email{fbouzaffour@ksu.edu.sa, mgarayev@ksu.edu.sa}
\maketitle%

\begin{abstract}We consider a realization of fractional supersymmetric of quantum mechanics of order $r$, where the Hamiltonian and supercharges involve reflection operators. It is shown that the Hamiltonian has $r$-fold degenerate spectrum and the eigenvalues of  hermitian supercharges are zeros of the associated Hermite polynomials of Askey and Wimp. Also it is shown that the associated eigenfunctions involve lacunary Hermite polynomials.
\end{abstract}

\maketitle

%
\section{Introduction}
Supersymmetry relates bosons and fermions on the basis of $\mathbb{Z}_2$-graded superalgebras \cite{Witten:1981nf, NicolaiXP}, where the fermionic set is realized by matrices of finite dimension or in terms of Grassmann variables .
In \cite{PlyushchayDC},  M. Plyushchay formulated a supersymmetric quantum mechanics (SUSYQM) for one-dimensional systems 
without introducing independent fermion ("spin-like") degrees of freedom but taking reflection operator $R$ as a grading operator, see also \cite{PlyushchayTY,CorreaJE,GamboaBZ}. In this realization, the  components of 
the supersymmetric Hamiltonian and supercharges are realized as differential
operators which additionally involve in their structure the reflection operator being a kind of Yang-Dunkl type operator \cite{PlyushchayTY,19,18}.
 In simplest supersymmetric systems of such a nature related to the quantum harmonic oscillator, the eigenstates are expressed in terms of Hermite orthogonal polynomials \cite{PostPE}.
\\ The fractional supersymmetric quantum mechanics of order $r$ is an extension of the ordinary supersymmetric quantum mechanics for
which the $\mathbb{Z}_2$-graded superalgebras are  replaced by  $\mathbb{Z}_r$-graded superalgebras
\cite{KhareGY,2,GamboaBZ,Daoud}. The framework of the fractional supersymmetric quantum mechanics has been shown to be quite fruitful. Among many works, we may quote the deformed Heisenberg algebra introduced in connection
with parafermionic and parabosonic systems \cite{PlyushchayTY,Daoud}, $C_\lambda$-extended oscillator algebra developed in the
framework of parasupersymmetric quantum mechanics \cite{QuesneCB}, and the generalized Weyl-Heisenberg
algebra $W_k$ related to $\mathbb{Z}_k-$graded supersymmetric quantum mechanics \cite{Daoud}. Note that the construction of fractional supersymmetric quantum mechanics without employment of fermions and parafermions degrees of freedom was started by Plyushchay \cite{PlyushchayDC, PlyushchayTY}. In particular, the idea of realization of fractional supersymmetry in the form as it was presented in refs. \cite{Daoud} and \cite{QuesneCB}, was proposed originally in \cite{GamboaBZ}, see also \cite{CorreaJE}.  \\

Due to the well-known fact that the  case $r=2$ corresponds to ordinary supersymmetry, it is interesting to study
which supersymmetric Hamiltonian involving reflection operators admits exact
 eigenfunctions which are expressible in terms of classical orthogonal polynomial system.
\\

In this paper, we give essentially an extension of the construction elaborated by Plyushchay in \cite{PlyushchayDC,PlyushchayTY} and develop a fractional supersymmetric quantum of order $r$ without parafermonic degrees of freedom. For this purpose, we use a generalization of Klein operator of the form $K^r=1,\,\,r\geq2,$ in order to study some mathematical properties of the fractional supersymmetric quantum mechanics. We also investigate such a possibility of realizing $r=2$ supersymmetry for the case of fractional supersymmetric qunatum mechanics.\\

The paper is organized as follows. In Section 2, we review the definition of the ordinary supersymmetric algebra  and provide some of their realizations. In section 3, we present a realization of the fractional supersymmetric quantum mechanics. In section 4, we construct a basis involving the lacunary Hermite polynomials that diagonalizes simultaneously  the Hamiltonian and the supercharge.  We show that the eigenvalues of the supercharge are the zeros of the associated Hermite polynomials.
\section{Ordinary supersymmetric quantum mechanics}
Let us briefly recall the formalism of the usual supersymmetric quantum mechanics. The SUSYQM, introduced by Witten \cite{Witten:1981nf},
is generated by three operators $Q_-,$ $Q_+$ and $H$ satisfying
\begin{equation}
Q_\pm^2=0,\quad[Q_\pm,H]=0,\quad \{Q_-,Q_+\}=H.\label{adam}
\end{equation}
Superalgebra (2.1) corresponds to a supersymmetric system with $N=2$ supersymmetry generated by
 two supercharges being  mutually conjugated nilpotent operators.
The usual construction of Witten's supersymmetric quantum mechanics with the superalgebra \eqref{adam}
is performed by introduction of fermion degrees of freedom (realized in a matrix form, or in terms of Grassmann variables) which commute with bosonic degrees of freedom \cite{ NicolaiXP}. Another realization of supersymmetric quantum mechanics is the so-called minimally  bosonized supersymmetric quantum. Its supercharges $Q_{i-}$ and $Q_{i+}$ have the expression (see, \cite{PlyushchayDC})
\begin{equation}\label{cha}
Q_{i-}=\frac{1}{\sqrt{2}}(\frac{d}{dx}+W(x))\Pi_i\quad \mbox{and}\quad Q_{i+}=\frac{1}{\sqrt{2}}\Pi_i(-\frac{d}{dx}+W(x)),\quad i=0,\,1,
\end{equation}
where $W(x)$ is an odd function $W(-x)=-W(x),$  $\Pi_0$ and $\Pi_1$ are the orthogonal projection operators defined by
\begin{equation}
\Pi_0=\frac{1}{2}(1+R),\quad \mbox{and}\quad \Pi_1=\frac{1}{2}(1-R).
\end{equation}
Here, $R$ is the reflection
operator
\begin{equation}
(Rf)(x)=f(-x).
\end{equation}
A straightforward computation shows that
\begin{equation}
Q_{i\pm}^2=0 ,\label{eq3}
\end{equation}
and the corresponding Hamiltonian $H_0$ and $H_1$ defined by \eqref{adam}, have the forms
\begin{align}
&H_0=-\frac{1}{2}\frac{d^2}{dx^2}+\frac{1}{2}W^2-\frac{1}{2}\,R\,W',\label{H11}\\&
H_1=-\frac{1}{2}\frac{d^2}{dx^2}+\frac{1}{2}W^2+\frac{1}{2}\,R\,W'.\label{H22}
\end{align}
In Eq. \eqref{cha}, two pairs of mutually conjugated operators correspond to two different $N=2$ supersymmeric systems with a bosonized supersymmetry. Each of such two systems can be obtained from the corresponding Witten's system \eqref{adam} with independent fermion degrees of freedom via a special non-local unitary transformation and subsequent reduction described in detail in ref. \cite{GamboaBZ,JakubskyKI}.  In addition, if the initial system \eqref{cha} is in the phase of unbroken supersymmetry, as it happens in the case of
  superoscillator, one of the corresponding resulting system will be described by unbroken bosonized $N=2$ supersymmetry,
  while another system will be described by the broken $N=2$ bosonized supersymmetry.  These two cases are described by Hamiltonian operators \eqref{HA1} and \eqref{HA2}, respectively.\\
In the particular case, when $W(x)=x$, the supercharges $Q_{i-}$ and $Q_{i+}$ are expressed as
\begin{align}\label{HH2}
Q_{i-}=a_-\Pi_i,\quad  Q_{i+}=\Pi_ia_+,
\end{align}
where  $a_+$ and $a_-$ are the standard creation and annihilation operators:
\begin{equation}\label{H4}
a_-=\frac{1}{\sqrt{2}}(\frac{d}{dx}+x),\quad a_+=\frac{1}{\sqrt{2}}(-\frac{d}{dx}+x).
\end{equation}
Recall that the eigenvalue problem corresponding to the harmonic oscillator is the differential equation for $\psi(x)$
\begin{equation}
-\frac{1}{2}\frac{d^2\psi(x)}{dx^2}+\frac{1}{2}\,x^2\psi(x)=\lambda \psi(x).\label{eq11}
\end{equation}
The wave functions  corresponding to the well-known eigenvalues
\begin{equation}
\lambda_n=n+\frac{1}{2},\quad n=0,\,1,\,2,\dots,\label{eq22}
\end{equation}
are given by
\begin{equation}
\psi_n(x)=(\sqrt{\pi} n!2^n)^{-1/2}e^{-x^2/2}H_n(x),
\end{equation}
where $H_n(x)$ is the Hermite polynomial of order $n.$
It follows that the wave functions $\psi_n(x)$ form an orthonormal basis in the Hilbert space $L^2(\mathbb{R})$ and  satisfy the orthogonality relations \begin{equation} \int_{\mathbb{R}}\psi_n(x)\psi_m(x)\,dx=\delta_{nm}.
 \end{equation}
From \eqref{H11} and \eqref{H22}, we have
\begin{align}\label{HA1}
&H_0=-\frac{1}{2}\frac{d^2}{dx^2}+\frac{1}{2}x^2-\frac{1}{2}R,\\&\label{HA2}
H_1=-\frac{1}{2}\frac{d^2}{dx^2}+\frac{1}{2}x^2+\frac{1}{2}R.
\end{align}
 It follows from equations \eqref{eq11} and \eqref{eq22} that
\begin{align}
&H_0\psi_{n}(x)=(n+\frac{1-(-1)^n}{2})\psi_{n}(x),\\&
H_1\psi_{n}(x)=(n+\frac{1+(-1)^n}{2})\psi_{n}(x).
\end{align}
The spectrum of $H_0$  consists only of the even numbers  starting with zero. Each level is degenerate except for the ground states which is unique and the spectrum of $H_1$ consists only of the odd numbers and each level is degenerate. In \cite{CorreaJE,PlyushchayQZ,CorreaPE}, the author  is used a very simple modification of supercharges \eqref{cha}
and Hamiltonian defined in \eqref{H11} and \eqref{H22},  and  showed   that  one can obtain a quantum system in the phase of the
broken supersymmetry (with all the doubled energy levels organized in  the
equidistant positive spectrum) described  by a linear superalgebra of the
form eqref{adam}. Moreover, in the same line one can obtain also the systems
described by nonlinear supersymmetric algebra; see  \cite{PlyushchayQZ,CorreaJE,CorreaPE}.\\

Instead  of the supercharges $Q_{i-}$ and $Q_{i+}$ defined
in \eqref{HH2}, we can also reformulate the algebra \eqref{adam} by
introducing a new hermitean operators $Q_i$  expressed in the form
\begin{equation}
\label{scharge}Q_i=\frac{1}{\sqrt{2}}\big(Q_{i-}+Q_{i+}\big),\quad i=0,\,1.\end{equation}
 Hermitian supercharges shown in \eqref{scharge}, represent a half of the corresponding pairs of the bosonized supercharges for the systems (2.14) and (2.15).
 From equations \eqref{HH2} and \eqref{H4} they 
share the property $Q^\dagger_{i}=Q_i$ and commute with the supersymmetric Hamiltonian $H_i,$ so it can be diagonalized  simultaneously with $H_i$. The supercharges $Q_i$ act on the wave functions as
\begin{align*}
&Q_0\psi_{2n}(x)=\sqrt{n}\,\psi_{2n-1}(x),\quad Q_0\psi_{2n+1}(x)=\sqrt{n+1}\,\psi_{2n+2}(x),\\&
Q_1\psi_{2n}(x)=\sqrt{n+1/2}\,\psi_{2n+1}(x),\quad Q_1\psi_{2n+1}(x)=\sqrt{n+1/2}\,\psi_{2n}(x).
\end{align*}
Put
\begin{align}\label{eig1}&\psi^{(0)}_{n\,\varepsilon}(x)=\frac{1}{2}(\psi_{2n+1}(x)
+\varepsilon\psi_{2n+2}(x)),\\&\label{eig2}
\psi^{(1)}_{n\,\varepsilon}(x)=\frac{1}{2}(\psi_{2n}(x)
+\varepsilon\psi_{2n+1}(x)),\quad \varepsilon=\pm 1.\end{align}
It is readily seen that
\begin{align*}
&Q_0\psi^{(0)}_{n\,\varepsilon}(x)=\varepsilon\sqrt{n+1}\psi^{(0)}_{n\,\varepsilon}(x),\\&
Q_1\psi^{(1)}_{n\,\varepsilon}(x)=\varepsilon\sqrt{n+1/2}\,\,\psi^{(1)}_{n\,\varepsilon}(x).
\end{align*}
Note that the eigenfunctions $\psi^{(i)}_{n\,\varepsilon}(x)$ of the supercharges $Q_i$ which are defined in equation \eqref{eig1} and \eqref{eig2} involve the lacunary Hermite polynomials $H_{2n}(x)$ and $H_{2n+1}(x)$.
\subsection{Associated Laguerre and Hermite polynomials}
The Hermite polynomials satisfy the three terms recurrence relation
\begin{equation}
\label{10}
\begin{split}
2xH_n(x)&= H_{n+1}(x)+2(n+1) H_{n-1}(x)\quad(n>0),\\
H_{-1}(x)&=0\quad  H_0(x)=1.
\end{split}
\end{equation}
The associated Hermite polynomials $H_n(x,c)$ are defined by the following three terms recurrence relation
\begin{equation}
\label{1}
\begin{split}
2xH_n(x,c)&= H_{n+1}(x,c)+2(n+1+c) H_{n-1}(x,c)\quad(n>0),\\
H_{-1}(x,c)&=0\quad  H_0(x,c)=1.
\end{split}
\end{equation}
They were introduced by Askey and Wimp in \cite{As}, where their weight functions and explicit formulas were also found.
When $c>-1$ , they satisfy the orthogonality relations
\begin{equation}
\int_{-\infty}^\infty \frac{H_n(x,c)H_m(x,c)}{|D_{-c}(ix\sqrt{2})|^2}\,dx=2^n\sqrt{\pi}\Gamma(n+c+1)\delta_{nm}.
\end{equation}
The function $D_{-c}(x)$ is a parabolic cylinder function
\begin{equation}
D_{2\nu}(2x)=2^\nu e^{-x^2}\Psi(-\nu,1/2;2x^2),
\end{equation}
and $\Psi(a,c;x)$ is the Tricomi  function \cite{Erd},
\begin{equation}
\Psi(a,c;x)=\frac{1}{\Gamma(a)}\int_0^\infty e^{-xt}t^{a-1}(1+t)^{c-a-1}\,dt,\quad \Re(a)>0,\,\Re(x)>0.
\end{equation}
Note that the polynomials $H_n(x,c)$ can be expressed in terms of  the associated Laguerre polynomials $L_n^{\nu}(x,c)$ and $\mathcal{L}_n^{\nu}(x,c)$,  which are defined by the three terms recurrence relations \cite{Ismail}
\begin{align}
&(2n+2c+\nu+1-x)L_n^{\nu}(x,c)=(n+c+1)L_{n+1}^{\nu}(x,c)+(n+c+\nu)L_{n-1}^{\nu}(x,c)\\&
L_0^{\nu}(x,c)=1,\quad L_1^{\nu}(x,c)=\frac{2c+\nu+1-x}{c+1}
\end{align}
and
\begin{align}
&(2n+2c+\nu+1-x)\mathcal{L}_n^{\nu}(x,c)=(n+c+1)\mathcal{L}_{n+1}^{\nu}(x,c)+(n+c+\nu)\mathcal{L}_{n-1}^{\nu}(x,c)\\&
\mathcal{L}_0^{\nu}(x,c)=1,\quad \mathcal{L}_1^{\nu}(x,c)=\frac{c+\nu+1-x}{c+1}
\end{align}
and have the orthogonality relations
\begin{equation}
\int_0^\infty L_n^{\nu}(x,c)L_m^{\nu}(x,c)\,x^\nu e^{-x}\frac{|\Psi(c,1-\nu;xe^{-i\pi})|^{-2}}{\Gamma(c+1)\Gamma(\nu+c+1)}\,dx
=\frac{(\nu+c+1)_n}{(\nu+1)_n}\delta_{nm}
\end{equation}
and \begin{equation}
\int_0^\infty \mathcal{L}_n^{\nu}(x,c)\mathcal{L}_m^{\nu}(x,c)\,x^\nu e^{-x}\frac{|\Psi(c,-\nu;xe^{-i\pi})|^{-2}}{\Gamma(c+1)\Gamma(\nu+c+1)}
\,dx=\frac{(\nu+c+1)_n}{(\nu+1)_n}\delta_{nm}.
\end{equation}
By \cite{Ismail}, we explicitly have
\begin{align}
&H_{2n}(x,c)=(-4)^n(1+c/2)_n L_n^{1/2}(x^2,c/2), \\& H_{2n+1}(x,c)=2x(-4)^n(1+c/2)_n\mathcal{L}_n^{-1/2}(x^2,c/2).
\end{align}
For example, the first polynomials are
\begin{align}
&H_0(x,c)=1,\,\,H_1(x,c)=2x,\\&
H_2(x,c)=4x^2-2(c+1),\,\,H_3(x,c)=8x^3-4(2c+3)x.
\end{align}
Furthermore,
\begin{equation}
H_n(x,c)=\sum_{k=0}^{[\tfrac{n}{2}]}\frac{(-2)^k(c)_k(n-k)!}{k!(n-2k)!}H_{n-2k}(x).
\end{equation}
The associated Laguerre and Hermite polynomials are also  birth and death processes polynomials, see Ismail, Letessier and Valent  \cite{Ismail}. \\

The lacunary Hermite polynomials $H_{nr+s}(x)$ may be defined by the following generating function \cite{Ges}:
\begin{equation}
G(r,s)=\sum_{n=0}^\infty H_{nr+s}(x)\frac{t^{nr+s}}{(nr+s)!},
\end{equation}
for arbitrary integers $r\geq 1$ and $0\leq s\leq r-1$. The specific  generating functions are
\begin{align*}
\\&G(1,0) =e^{-t^2+2xt}~,  \\&
G(2,0) = e^{-t^2} \cosh(2xt),  \\&
G(2,1) = e^{-t^2} \sinh(2xt).
\end{align*}
Note that  $G(1,0)$ is the generating function of the Hermite polynomials, $G(2,0)$ and $G(2,1)$ are  linear combinations of $G(1,0)$
(the same function with $x$ changed by $-x$). Physically, these generating functions are very useful
 to describe the decomposition of ordinary, even and odd coherent states \cite{Nie}.
They also appeared in a number of problems, including for example the treatment of Cauchy problems in partial differential equations
 \cite{Pen,Dab} and in calculations involving coherent and squeezed states \cite{Dab}.\indent
\section{Realization of fractional supersymmetric quantum mechanics}
Following Khare \cite{KhareGY}, a fractional supersymmetric quantum mechanics model of arbitrary order $r$ can be developed by generalizing the fundamental equations \eqref{adam} to the forms
\begin{align}
& Q_\pm^{r}=0,\quad [H,Q_\pm]=0, \quad Q_-^\dagger=Q_+,\\&
Q_{-}^{r-1}Q_{+}+Q_{-}^{r-2}Q_+ Q_-+\dots+ Q_-Q_+ Q_{-}^{r-2}+Q_+ Q_{-}^{r-1}=(r-1)Q_{-}^{r-2}H.
\end{align}
It was shown in \cite{2} that the parasupercharges $Q_-$ and $Q_+$ can be represented
by $r \times r$ matrices as natural extensions of the $r = 2$ scheme. Note that the fractional supersymmetric quantum mechanics was also developed without an explicit introduction of parafermonic degrees of freedom by Plyushchay \cite{PlyushchayDC}, see also \cite{PlyushchayTY,KlishevichNQ}.\\

We shall now present a construction of fractional supersymmetric quantum mechanics of order $r$ by using a Klein type operator. Let $\mathcal{F}$ be  the Fock space of states for the ordinary bosonic oscillator generated  by the complete set of orthonormal vectors $\{|n\rangle\}_{n=0}^\infty$:
\begin{equation}|n\rangle=\frac{1}{\sqrt{n!}}a_+^n|0\rangle,\quad n=1,\,2,\,\dots,\label{r1}\end{equation}
constructed over the vacuum state defined by the relation,
$$a_-|0\rangle=0.$$
The  creation and annihilation  operators   $a_+$ and $a_-$ may be realized
in  Fock space  as
\begin{align}\label{repp}
&a_- |n>= \sqrt{n}|n-1>\quad \mbox{and}\quad a_+|n\rangle=\sqrt{n+1}|n+1>,
\end{align}
The number operator $N$, which counts the number of
particles, or objects in general in a system, is defined by $N = a_+ a_-$ and is represented as
\begin{equation*}
N|n\rangle= n|n\rangle.
\end{equation*}
  It satisfies the commutation
relations
\begin{align*}
[a_-, N]=a_-, \quad [a_+, N]=-a_+.
\end{align*}
Now, let $r\geq 2$ be a fixed integer. We introduce the generalized Klein operator $K$ defined by the relations (see, \cite{Bouz2,Bouz3})
\begin{equation}
K^{r}=1,\quad a_-K=\varepsilon_{r} Ka_-
\quad a_+ K=\varepsilon_{r}^{-1} Ka_+,\quad K^\dagger=K^{-1},
\end{equation}
where
$\varepsilon_r=e^{\frac{2i\pi}{r}}$ is a primitive root of unity. The operator $K$ can be realized by using the number operator $N$ in the form $K=e^{\frac{2i\pi }{r}N},$ and acts on  the  space $\mathcal{F}$ as
\begin{equation}
K\mid n\rangle=\varepsilon_r ^n \mid n\rangle, \end{equation}
and so introduces
$\mathbb{Z}_r$-grading structure on the Fock space $\mathcal{F}$  as
\begin{equation}
\mathcal{F}=\bigoplus_{j=0}^{r-1}\mathcal{F}_j,
\end{equation}
where
$$\mathcal{F}_j=\{\mid nr+j\rangle:n=0,\,1,\,\dots\}.$$
For $j=0,\,1,\,\dots,\,r-1,$ we denote by $\Pi_{j},$ the orthogonal projection from $\mathcal{F}$ onto its subspace $\mathcal{F}_j$, which can be represented as
\begin{equation}
\Pi_{j}=\frac{1}{r}\sum_{l=0}^{r-1}\varepsilon_r^{-lj}K^l.
\end{equation}
Equivalently, the operators $K^l$ are expressed in terms of $\Pi_{j}$ as follows
\begin{equation}
K^l=\sum_{j=0}^{r-1}\varepsilon_r^{-lj}\Pi_{j},\quad l=0,\dots,r-1.
\end{equation}
It is clear that they form a system of resolution of the identity:
\begin{equation}\label{or1}
\Pi_0+\Pi_1+\dots+\Pi_{r-1}=1\quad \mbox{and}\quad \Pi_i\Pi_j=\delta_{ij}\Pi_i.
\end{equation}
The action of $\Pi_j$ on $\mathcal{F}$ can be taken to be
\begin{equation}
\Pi_j\mid kr+l\rangle=\delta_{jl}\mid kr+l\rangle.
\end{equation}
\subsection{Supercharges}Following \cite{PlyushchayDC,PlyushchayTY,Daoud}, for every fixed integer $j$ such that $0\leq j \leq r-1,$  one can construct  two supercharges $Q_{j-}$ and $Q_{j+}$ by means of the orthogonal projections $\Pi_k$ and the creation  and annihilation  operators $a_+$ and  $a_-$ introduced before as follows
\begin{align}\label{Q1}
&Q_{j-}=a_-(1-\Pi_j)=(1-\Pi_{j-1})a_-,\\& \label{Q2}Q_{j+}=a_+(1- \Pi_{j-1})=(1-\Pi_{j})a_+.
\end{align}
Obviously, the operators $a_-,$  $a_+$ , $Q_{j-}$ and $Q_{j+}$ satisfy the intertwining relations
\begin{align}\label{i1}
&\Pi_{j-1} a_-=a_-\Pi_{j}\quad \mbox{and} \quad
a_+\Pi_{j-1} =\Pi_{j}a_+\\&\label{i2}Q_{j-}a_-=a_-Q_{j+1\,-}\quad \mbox{and}\quad a_+Q_{j+}=Q_{j+1\,+}a_+.
\end{align}It is easy to see that the supercharges $Q_{j-}$ and $Q_{j+}$ have the Hermitian conjugation relations
\begin{equation}
Q_{j-}^\dagger=Q_{j+}.
\end{equation}
Furthermore, the operators $a_-$ and $a_+$ can be represented as
\begin{equation}a_-=\frac{1}{r-1}\sum_{j=0}^{r-1}Q_{j\,-}\quad \mbox{and} \quad a_+=\frac{1}{r-1}\sum_{j=0}^{r-1}Q_{j\,+}.\end{equation}
Now, by making use of  relations \eqref{i1} and \eqref{i2},  we easily obtain
$$ Q_{j-}^{k} = \begin{cases}a_{-}^{k}\big(1-\Pi_j-\dots\, \Pi_{j+k}\big) & \mbox{ if }\, k\leq r-1-j\\ a_{-}^{k}\big(1-\Pi_j-\dots\, \Pi_{N-1}-\dots-\Pi_{k+j-r-1}\big) &\mbox{ if }\, r-j\leq k\leq r-1.\end{cases} $$
In particular, for $k=r$, we get
\begin{align}
Q_{j-}^r&=a^r\big(1-\Pi_j-\dots\, \Pi_{r-1}-\dots-\Pi_{j-1}\big)=0.
\end{align}
Similarly, we have $$Q_{j+}^r=0.$$ Hence, the supercharge operators $Q_{j\pm}$  are nilpotent operators of order $r.$
\subsection{Supersymmetric Hamiltonian of order $r$}
We can use the operators $a_-$ , $a_+$ and the orthogonal projections $\Pi_0,\,\Pi_1,\,\dots,\Pi_{r-1}$ introduced before for obtaining a realization of supersymmetric Hamiltonian.\\
The first main result, is
\begin{theorem} For fixed $j=0,\,\dots,\, r-1,$  the hermitian operators $Q_{j-},$ $Q_{j-}$ and $H_j$ given by
\begin{align*}
&Q_{j+}=a_+(1-\Pi_{j-1}),\quad  Q_{j-}=a_-(1-\Pi_{j}),\\&
H_j=a_+a_-+\sum_{k=0}^{r-1}(1+k-r/2)\Pi_{r+j-k-1}
\end{align*}
satisfy the commutation relations
\begin{align*}
& Q_{j\pm}^{r}=0,\quad [H,Q_{j\pm}]=0, \quad Q_{j-}^\dagger=Q_{j+},\\&
Q_{j-}^{r-1}Q_{j+}+Q_{j-}^{r-2}Q_{j+} Q_{j-}+\dots+ Q_{j-}Q_{j+} Q_{j-}^{r-2}+Q_{j+} Q_{j-}^{r-1}=(r-1)Q_{j-}^{r-2}H_j.
\end{align*}
Furthermore, the superoscillator  $H_j$ has $r$-fold degenerate spectrum and acts on the occupation basis as
\begin{equation*}
H_j\mid nr+j+s\rangle=(rn+r/2+j)\mid nr+j+s\rangle.
\end{equation*}
\end{theorem}
\begin{proof}
 We  first summarize the key properties \cite{Daoud}:
 \begin{align}\label{d1}
&Q_{j+} Q_{j-}^{r-1}=a_+ a_{-}^{r-1}\Pi_{j-2},\\&\label{d2} Q_{j-}^{r-1}Q_{j+} =a_{-}^{r-1} a_+\Pi_{j-1},\\&\label{d3}Q_{j-}^{r-1-k} Q_{j+}Q_{j+}^{k}=a_{-}^{r-1-k}a_+ a_{-}^k(\Pi_{j-2}+\Pi_{j-1}),\quad k=1,\dots,r-2.
\end{align}
From the  relation $a_{-}^ka_+=a_+a_{-}^{k}+ka^{k-1}$ and relations \eqref{d1}, \eqref{d2} and \eqref{d3}, we have
\begin{align}
\sum_{k=0}^{r-1}Q_{j-}^{r-1-k}Q_{j+} Q_{j-}^k&=(r-1)Q_{j-}^{r-2}\big[a_+a_-+
(2-\frac{r}{2})\Pi_{j-2}(r)+(1-\frac{r}{2})\Pi_{j-1}\big].
\end{align}
On the other hand,  we get from \eqref{or1} that
\begin{equation}
Q_{j-}^{r-2}\sum_{k=2}^{r-1}(1+k-\tfrac{r}{2})\,\Pi_{r+j-k-1}=0.
\end{equation}
This yields
\begin{align}
\sum_{k=0}^{r-1}Q_{j-}^{r-1-k}Q_{j+} Q_{j-}^k&=(r-1)Q_{j-}^{r-2}H_j,
\end{align}
where the operator $H_j$ is given explicitly by
\begin{align}
&H_j=a_+ a_-+\sum_{k=0}^{r-1}(1+k-\tfrac{r}{2})\,\Pi_{r+j-k-1}.
\end{align}
It remains to prove that
\begin{align}
[H,Q_{j-}]=[H,Q_{j+}]=0.
\end{align}
Indeed, from \eqref{d3}, we have
\begin{align*}
Q_{j-}H&=a_{-}a_+a_-(1-\Pi_{j})+\sum_{k=0}^{r-2}
(1+k-\frac{r}{2})\,a_-(1-\Pi_{j})\Pi_{r+j-k-1}
\\&=a_{+}a_-a_-(1-\Pi_{j})+\sum_{k=0}^{r-2}
(2+k-\frac{r}{2})\,\Pi_{r+j-k-2}\big)(1-\Pi_{j-1})a_-
\\&=a_{+}a_-a_-(1-\Pi_{j})+\sum_{k=1}^{r-1}
(1+k-\tfrac{r}{2})\,\Pi_{r+j-k-1}(1-\Pi_{j-1})a_-\\&
=HQ_{j-}.
\end{align*}
Thus,
\begin{equation} \label{dd1}
[H_j,Q_{j-}]=0.
\end{equation}
Since the orthogonal projection $\Pi_k$ is hermitian conjugate $\Pi_k^\dagger=\Pi_k$, we easily check that the operators $H_j$ share the property
\begin{equation}\label{h2}
H_j^\dagger=H_j.
\end{equation}
Combining this with \eqref{dd1}, we get
\begin{equation}
[H_j,Q_{j+}]=0
\end{equation}
Given representation \eqref{repp}, 
the action of operators $Q_{j-}$ and $Q_{j+}$ reduces  to
\begin{align}
&Q_{j-}|nr+j+s\rangle=\sqrt{nr+j+s}|nr+j+s-1\rangle,\,\, s=1,\dots,r-1,\\& Q_{j+}|nr+j+s\rangle=\sqrt{nr+j+s+1}|nr+j+s+1\rangle,\,\,s=0,\dots,r-2,\\&
Q_{j-}|nr+j\rangle=0,\quad Q_{j+}|(n+1)r+j-1\rangle=0
\end{align}
and hence,
\begin{equation}
H_j\mid nr+j+s\rangle=(rn+r/2+j)\mid nr+j+s\rangle.
\end{equation}
\end{proof}
\section{Simultaneous diagonalization}
Our second   result is the following.
\begin{theorem}For fixed integer $j=0,\,\dots,\,r-1,$ the operator $Q_j$ defined by  \begin{equation}Q_j=\frac{1}{\sqrt{2}}(Q_{j-}+Q_{j+} ),\end{equation}is self-adjoint and commutes with the oscillator $H_j,$ that is
\begin{equation}
Q_j^\dagger=Q_j,\quad [Q_j,H_j]=0.
\end{equation}
Furthermore,  for every integer $n$, the operator $Q_j$  leaves  the subspace spanned  by the states $\{\, |nr+j\rangle,\,\dots,\,
|(n+1)r+j-1\rangle\}$
  invariant and restricted operator  has  a nondegenerate spectrum
  which consists from the roots of the associated Hermite polynomials $H_r(x,c),$ with $c=nr+j-1.$
\end{theorem}
\begin{proof}
We have:
\begin{align}\label{Re1}
&Q_j|nr+j+s\rangle=a_s(r,n,j)|nr+j+s-1\rangle+a_{s+1}(r,n,j)|nr+j+s+1\rangle,\,\, \\&\label{Re2}
Q_j(|nr+j\rangle=a_1(r,n,j)|nr+j+1\rangle,\\&\label{Re3} Q_j|(n+1)r+j-1\rangle=a_{r-1}(r,n,j)|(n+1)r-2\rangle,\end{align}
where $$a_s(r,n,j):=\sqrt{(nr+j+s) /2},\quad s=1,\dots,\,r-1.$$
From \eqref{Re1}, \eqref{Re2} and \eqref{Re3}, we see that the operator $Q_j$ leaves the subspace generated by the states $\mid nr+j+s\rangle, s=0,\,1,\,\dots,\,r-1,$ invariant. Hence, it can be
represented in this basis by the following $r\times r$ tridiagonal
Jacobi matrix $A_{n\,r}$:
\[
A_{n\,r} =
 \begin{pmatrix}
  0 & a_1(r,n,j)& 0 &    \\
  a_{1}(r,n,j) & 0 & a_2(r,n,j) & 0  \\
   0  &  a_2 (r,n,j)& 0 & a_3(r,n,j)    \\
   &   &  \ddots &    \ddots  \\
      & & \dots &a_{r-2} (r,n,j)&0 & a_{r-1} (r,n,j)\\
            & & \dots &0 & a_{r-1} (r,n,j)& 0
\end{pmatrix}.
\]
It is well known that if the coefficients of the subdiagonal of some Jacobi matrix are different from zero \cite{Chihara}, then all the eigenvalues of this matrix are real and nondegenerate.
Introduce the normalized
eigenvectors $\mid n,j,\,s\rangle $  of the supercharge $Q_{j}$:
\begin{equation}
Q_j\mid n, j,\,s\rangle=x_{s}(r,n,j)
\mid n,j,\,s\rangle,\,\,s=0,\,\dots,r-1.
\end{equation}
The eigenvalues $x_{s}(r,n,j)$ are real and nondegenerate:
$$x_{s}(r,n,j)\neq x_{t}(r,n,j)\quad \mbox{if}\quad s\neq t.$$
The vector $\mid n, j,s\rangle $ can be expanded in the basis  $\mid nr+j+s\rangle, \,\,s=0,\,1,\,\dots,\,r-1,$
\begin{equation}\label{w1}
\mid n, j,\,s\rangle=\sum_{k=0}^{r-1}\Upsilon_{sk}|nr+j+k\rangle.
\end{equation}
The coefficient $\Upsilon_{sk}$ can be rewritten :
\begin{equation}
\Upsilon_{sk}=\sqrt{w_s}P_{k}(x_{s}(r,n,j)).\label{111}
\end{equation}
By construction, we have $$P_0(x_s(r,n.j))=1\quad \text{and}\quad \sqrt{w_s}=\Upsilon_{s0}.$$
Relation \eqref{w1} becomes \begin{equation}\label{w12}
\mid n, j,\,s\rangle=\sum_{k=0}^{r-1}\sqrt{w_s}P_{k}(x_{s}(r,n,j))|nr+j+k\rangle.
\end{equation}
It follows from \eqref{Re1}, \eqref{Re2} and \eqref{Re3}, that the coefficients $P_k(x)$ obey
the three term recurrence relation
\begin{align}a_{k}(r,n,j)\,P_{k-1}(x)+a_{k+1}(r,n,j)\,P_{k+1}(x)=x P_k(x),\quad k=0,\dots r-1,
\end{align}
and hence are orthogonal polynomials with the initial condition $P_{-1}(x) = 0$ . Since  the occupation basis is  orthonormal,
\begin{equation}\label{w112}
\sum_{k=0}^{r-1}w_sP_{k}(x_{s}(r,n,j))P_{l}(x_{s}(r,n,j))=\delta_{kl},
\end{equation}
and this provides the following inverse relation of \eqref{w12}:
\begin{equation}\label{w1112}
|nr+j+k\rangle=\sum_{s=0}^{r-1}\sqrt{w_s}P_{k}(x_{s}(r,n,j))\mid n, j,\,s\rangle.
\end{equation} The monic orthogonal polynomials $\widetilde{P}_k(x),$ related to $P_{k}(x),$ satisfy the three term recurrence relation
\begin{align}\label{r22}&x \widetilde{P}_{k}(x)=\widetilde{P}_{k+1}(x)+(nr+j+k)/2\, \widetilde{P}_{k-1}(x)\,\,k=0,\,\dots,\,r-1,\\&
 \widetilde{P}_{-1}(x)=0, \quad \widetilde{P}_0(x)=1.\label{r222}
\end{align}
From the three terms recurrence relations \eqref{1}, we see that the polynomials $\widetilde{P}_{k}(x)$ satisfy  \eqref{r22}
and \eqref{r222} and  can be identified with the associated Hermite polynomial $H_k(x,c),$ with $c=nr+j-1.$ Namely, we have
\begin{equation}
\widetilde{P}_k(x)=2^{-k}H_k(x,nr+j-1),\quad k=0,\,1,\dots,r-1.
\end{equation}
It is well-known from the
theory of orthogonal polynomials \cite{Ismail, Chihara} that the eigenvalues $x_{s}(r,n,j)$ of the Jacobi matrix $A_{n\,r}$ coincide with the roots of the characteristic  polynomial $\widetilde{P}_r(x)$ given by \begin{equation}
\widetilde{P}_r(x)=(x-x_{0}(r,n,j))\dots(x-x_{r-1}(r,n,j))
\end{equation}
and that the discrete weights $w_s$ can be expressed by (see, \cite{Ismail})
\begin{equation}
w_s=\frac{\xi_{r-1}}{H_{r-1}(x_{s}(r,n,j),nr+j-1)H'_r(x_{s}(r,n,j),nr+j-1)},
\quad s=0,1,\dots,r-1,\end{equation}
where the normalization constants are given by
\begin{equation}
\xi_{r-1}=2^{r}\frac{\Gamma(r+nr+j-1)}{\Gamma(nr+j)}.
\end{equation}
\end{proof}

\section{Conclusion}
We considered  a fractional supersymmetric quantum mechanics of finite order  having a structure similar
to one constructed by Pluyshchay \cite{PlyushchayDC,PlyushchayTY}. The Hamiltonian and the supercharges involve  reflection operator. We showed this system to be exactly solvable and found that its wave functions are expressed in terms of lacunary Hermite polynomials and associated Hermite polynomials.

\end{document}